\documentclass[12pt,reqno,a4paper]{amsart}

\usepackage[margin=3cm,footskip=1cm]{geometry}

\usepackage{amsmath}
\usepackage{amsfonts}
\usepackage{amsthm}
\usepackage{amstext}
\usepackage{enumerate}
\usepackage[latin1]{inputenc}
\usepackage[shortlabels]{enumitem}

\usepackage{amssymb}
\usepackage{latexsym}

\usepackage{graphicx}
\usepackage{verbatim}

     \newcommand{\supp}{\operatorname{supp}}

     \newcommand{\R}{{\mathbb{R}}}

\newcommand{\e}{{\rm e}}
\newcommand{\ess}{{\rm ess}}

\renewcommand{\i}{{\rm i}}
\renewcommand{\d}{{\rm d}}

\renewcommand{\Re}{{\rm Re}\,}

\newcommand\inp[2][]{#1 \langle #2#1\rangle}
\newcommand\parb[2][]{#1 \big ( #2#1\big )}

\newcommand{\pp}{{\rm pp}}

\newcommand{\mand}{\text{ and }}
\newcommand{\mfor}{\text{ for }}

\newcommand{\vD}{{\mathcal D}}

\newcommand{\vH}{{\mathcal H}}

\newcommand {\tk}{{\theta_m}}
\newcommand {\tkp}{{\theta_m^\prime}}

\newcommand{\Tk}{{\Theta_m}}
\newcommand{\Tkp}{{\Theta_m^\prime}}

     \theoremstyle{plain}
     \newtheorem{thm}{Theorem}[section]
     \newtheorem{proposition}[thm]{Proposition}
     \newtheorem{lemma}[thm]{Lemma}
      \newtheorem{corollary}[thm]{Corollary}
     \theoremstyle{definition}

\newtheorem{examples}[thm]{Examples}
     
     \newtheorem{cond}[thm]{Condition}

     \newtheorem{remarks}[thm]{Remarks}

\newtheorem*{remarks*}{Remarks}
\newtheorem*{remark*}{Remark}

     \numberwithin{equation}{section}

\title[Absence of eigenvalues]{Absence of embedded eigenvalues for Riemannian Laplacians}

\thanks{
This work was essentially done during K.I.'s stay in Aarhus University (academic
year 2009-2010).
He would like to express his gratitude for financial support from  FNU 160377
(2009--2011) as
well as  from  JSPS Wakate (B) 21740090
(2009--2012). K.I. thanks H. Kumura for valuable discussion on this topic.}

\author{K. Ito}
\address[K. Ito]{Graduate School of Pure and Applied Sciences, University of Tsukuba\\
1-1-1 Tennodai, Tsukuba Ibaraki, 305-8571 Japan}
\email{ito-ken@math.tsukuba.ac.jp}

\author{E. Skibsted}
\address[E. Skibsted]{Institut for  Matematiske
Fag \\
Aarhus Universitet\\ Ny Munkegade  8000 Aarhus C,
Denmark}
\email{skibsted@imf.au.dk}

\begin{document}
\begin{abstract}
In this paper we study absence of embedded eigenvalues
 for Schr\"odinger operators on non-compact 
connected  Riemannian manifolds.  A   principal example is given by a manifold with an 
  end  (possibly more than one) in
which   geodesic coordinates are naturally defined. In this case one of our
geometric  conditions  is  a  positive lower bound of the second fundamental form of
angular submanifolds at infinity inside  the end. Another condition may be viewed 
   (at least in a special case) as  being a bound  of the trace of
this quantity, while similarly, a third one as being a bound of the
derivative of this trace. In addition to  geometric bounds we need
conditions on the potential, a regularity
property of the domain of the Schr\"odinger operator and the unique
continuation property. Examples include ends endowed with asymptotic Euclidean or hyperbolic
metrics studied  previously in the literature.
\end{abstract}
\maketitle
\tableofcontents

\section{Introduction and results}\label{sec:introduction}
Let $(M,g)$ be a non-compact connected  
  Riemannian manifold of dimension $d\ge 1$  (possibly
 incomplete),
and $H$ the Schr\"odinger operator
on the Hilbert space ${\mathcal H}=L^2(M)$:
\begin{align*}
H=H_0+V;\quad H_0=-\tfrac12\triangle=\tfrac12p_i^*g^{ij}p_j,\quad p_i=-\mathrm{i}\partial_i.
\end{align*}
We  introduce four conditions  under which we  prove that a
self-adjoint realization of $H$ does
not have eigenvalues greater than some computable constant. For    the
Euclidean case the theory boils down to 
absence of positive eigenvalues which is  a  well studied
subject, see e.g.  \cite{RS,FHH2O,JK}. Our conditions appear  rather
weak and allow for application to  manifolds with boundary (possibly
caused by metric or potential singularities). In
particular, to our knowledge,  
they are  weaker than conditions used so far in  the literature  on
the subject,
cf.  e.g. \cite{Me, MZ, Do, Ku, Ku2}. The present work is applied in a
companion paper  \cite{IS}  in which scattering
theory  is studied for a  general class of metrics. Our conditions are
also weaker than the
 conditions of \cite{IS}.

The first condition we impose guarantees intuitively that 
$(M,g)$ has at least one ``expanding end''. 
\begin{cond}\label{cond:diffeo} 
There exists an unbounded real-valued function $r\in C^\infty(M)$,
$r(x)\geq 1$, such that uniformly in $x\in M$ (i.e. all limits below
are meant to be  uniform in $x\in M$):
\begin{enumerate}[\normalfont (1)]
\item The following inequality holds, 
\begin{align}
\limsup_{r\to\infty} {}|\mathrm{d} r|<\infty.\label{eq:3}
\end{align}
\item There exist  constants $c>0$, $\tilde c\in [c/2,c)$  and $r_0\geq 1$ such that 
\begin{align}
\nabla^2 r^2 \ge cg \mfor r\geq r_0,\label{eq:10.2.8.3.55}
\end{align}
and 
\begin{align}
\liminf_{r\to\infty}{}(r\partial^r|\mathrm{d} r|^2+\tilde c    |\mathrm{d} r|^2)>  0,\quad 
\lim_{r\to\infty}{}\partial^r|\mathrm{d} r|^2= 0,\label{eq:11.7.15.8.1}
\end{align}
where $\partial^r=\i p^r=\nabla r=\mathop{\mathrm{grad}} r$ denotes the gradient vector field for $r$, i.e.
\begin{align*}
\partial^rf=(\partial_i r)g^{ij}(\partial_j f),\quad f\in C^\infty(M).
\end{align*} 
\item \label{item:11.7.15.8.2}
There exists a decomposition $\triangle r^2= \rho_1+\rho_2$ such that 
\begin{align}
\lim_{r\to\infty}\rho_1=0,
\quad\limsup_{r\to\infty}r^{-1}|\rho_2|<\infty,
\quad
\limsup_{r\to\infty}{}|\mathrm{d}\rho_2|<\infty.
\label{eq:10.9.2.23.19}
\end{align}
\end{enumerate}
\end{cond}


Note that the subsets $\{x\in M|\ r(x)\le \tilde r\}$, $\tilde r\ge 1$,
may not be compact (this is similar  to  \cite{Ku, Ku2}, see Subsection
 \ref{sec:conditions-inside-an}). In particular the
function $r$ could model a distance
function within a fixed single  {\it end} of $M$ extended to be  bounded
outside, in particular bounded in 
other  ends  of   $M$. Also note that for an exact distance function \eqref{eq:3} and
 \eqref{eq:11.7.15.8.1} are  trivially fulfilled, and in that
 case the above operator $\partial^r$ is
 identified as the geodesic radial derivative $\partial_r$,  see Subsection
 \ref{sec:conditions-inside-an}.

\begin{cond}\label{cond:10.6.1.16.24}
There exists a decomposition $V=V_1+V_2$, $V_1\in L^2_{\rm loc}(M)$,  $V_2\in C^1(M)$ and $V_1,V_2$
                                real-valued, such that uniformly in $x\in M$: 

\begin{align}\label{eq:60}
\lim_{r\to\infty} rV_1=0,\quad \limsup_{r\to \infty} |V_2|<\infty ,\quad
\limsup_{r\to\infty}{}r \partial^r V_2\le 0.
\end{align}

\end{cond}

Note that under Condition \ref{cond:10.6.1.16.24} the subspace $C^\infty_{\mathrm{c}}(M)\subseteq  {\mathcal D}(V)$ and
whence that $H$ is defined at least on
$C^\infty_{\mathrm{c}}(M)$. However  under
Conditions~ \ref{cond:diffeo} and \ref{cond:10.6.1.16.24} this operator
 is  
 not necessarily essentially self-adjoint.
Note that $(M,g)$ is allowed to be incomplete and that $V$ is  allowed to
be unbounded. For instance 
$(M,g)$ could be the interior of a Riemannian  manifold     with  boundary  and for essentially self-adjointness we
would then need a symmetric boundary condition.
Lack of essential self-adjointness could also originate  from  
unboundedness of $V$ in some end. 
 To  fix a self-adjoint extension we first choose a non-negative $\chi\in C^\infty(\mathbb{R})$ with
\begin{align*}
\chi(r)=\left\{\begin{array}{ll}
0&\mbox{ for } r\le 1, \\
1 &\mbox{ for }r \ge 2,
\end{array}
\right.
\end{align*}
and then set 
\begin{align}
\chi_\nu(r)=\chi(r/\nu),\quad  \nu\ge 1.
\label{eq:11.7.11.5.14}
\end{align}
We shall henceforth consider the function $\chi_\nu$ as being  composed with the
function $r$ from Condition \ref{cond:diffeo}. In this sense
particularly   
$\chi_\nu\in C^\infty (M)$.
\begin{cond}\label{cond:11.7.11.1.25}
The operator $H$ defined on $C^\infty_{\mathrm{c}}(M)$ (by 
Condition \ref{cond:10.6.1.16.24}) has a self-adjoint extension,
denoted by $H$ again,
such that for any $\psi\in {\mathcal D}(H)$ there exists a sequence 
$\psi_n\in C^\infty_{\mathrm{c}}(M)$ such that for all large $\nu\ge 1$
\begin{align*}
\|\chi_\nu (\psi-\psi_n)\|+\|\chi_\nu (H\psi - H\psi_n)\|\to 0\quad \mbox{ as }n\to\infty.
\end{align*}
\end{cond}
Note that Condition \ref{cond:11.7.11.1.25} is fulfilled if  $(M,g)$ is
 complete and   $V$ is
bounded. In that case indeed $H$ is essentially
self-adjoint on
$C^\infty_{\mathrm{c}}(M)$, see Proposition
\ref{prop:global-conditionss} for a more general result.

As a global condition we impose   for this self-adjoint extension   \textit{the unique continuation   property}.
\begin{cond}\label{cond:11.7.9.0.24}
If $\phi\in {\mathcal D}(H)$ satisfies $H\phi=E\phi$, $E\in \mathbb{R}$, and $\phi(x)=0$ in some open subset, 
then $\phi(x)=0$ in $M$.
\end{cond}

In Section \ref{sec:11.7.9.0.25} we shall discuss various models satisfying 
Conditions~\ref{cond:diffeo}--\ref{cond:11.7.9.0.24}.
 We define a ``critical'' energy,
 \begin{equation}
   \label{eq:E_0}
   E_0=\limsup_{r\to\infty}\parb{V+\tfrac{|\d \rho_2|^2}{32(c-\tilde c)\tilde c}}.
 \end{equation} Note that the smallest possible value of $E_0$ under
 variation of $\tilde c$ in \eqref{eq:11.7.15.8.1} is attained at
 $\tilde c=c/2$. For   examples in Subsection
 \ref{sec:conditions-inside-an} (for which for simplicity $V=0$) we can use this $\tilde c$  and
 verify that
 the essential spectrum $\sigma_\ess (H_0)=[E_0, \infty)$, see Remark
 \ref{remark:end-warped-products} \ref{item:6}. Whence   
  for these examples indeed  $ E_0$   is
 critical regarding absence of eigenvalues as stated more generally in the following theorem.

\begin{thm}\label{thm:absence-eigenvalues-1}
Suppose Conditions~\ref{cond:diffeo}--\ref{cond:11.7.9.0.24}.
Then the eigenvalues  of  $H$ are  absent above $E_0$, i.e. 
$\sigma_{\mathrm{pp}}(H)\cap (E_0,\infty)=\emptyset$.
\end{thm}

 Various of our conditions are optimal for exclusion of embedded
 eigenvalues. It is well known in  Schr\"odinger operator theory that the  von Neumann Wigner potential, see for example
 \cite{FH} or \cite[Section XIII.3]{RS}, provides an example of a
 positive eigenvalue for  a decaying potential $O(r^{-1})$, $r=|x|$. Whence  the conclusion of Theorem
 \ref{thm:absence-eigenvalues-1} is in general false if the first
 condition of \eqref{eq:60} is relaxed as $\limsup_{r\to
   \infty}r|V_1|<\infty$. An example of a 
 Laplace-Beltrami operator having an embedded eigenvalue is 
 constructed in \cite{Ku}. This is  for a hyperbolic metric, and the
 example shows similarly that  the conclusion of Theorem
 \ref{thm:absence-eigenvalues-1}  in general is false if the first
 condition of \eqref{eq:10.9.2.23.19} is relaxed as $\limsup_{r\to
   \infty}|\rho_1|<\infty$. (Actually Kumura uses the  von Neumann
 Wigner potential in his construction.) 

The proof of Theorem~\ref{thm:absence-eigenvalues-1} 
follows the scheme of  \cite{FHH2O,FH,DG,MS} 
employing  in particular a Mourre-type commutator estimate and exponential decay estimates of a priori
   eigenstates. In our geometric setting  the ``Mourre commutator'' can be very singular
   (in particular  not bounded relatively to $H$ in any usual sense).
   Consequently we only have a weak (however sufficient)  version of the  commutator estimate, see
   Corollary \ref{cor:10.10.13.15.00}.

\smallskip

We use throughout the paper the standard notation
$\inp{\sigma}=(1+|\sigma|^2)^{1/2}$ and (as above) $\d$ for exterior
differentiation (acting on functions on $M$). Note that in local
coordinates $p:=-\i \d$ takes the form $p=(p_1,\dots, p_d)$. We shall
slightly abuse notation writing for example $p\psi\in {\mathcal H}=L^2(M)$
for $\psi\in C^\infty_{\mathrm{c}}(M)$ even though the correct meaning 
here  is a section of the (complexified) cotangent bundle, i.e. $p\psi\in
\Gamma(T^*M)$. Note at this point that
$\|p\psi\|:=\|p\psi\|_{\Gamma(T^*M)}=\|\,|p\psi|\,\|_{\mathcal H}$. If
$A$ is an operator on $\vH$ and $\psi\in \vD(A)$ we denote the
expectation $\inp {\psi,A\psi}$ by $\inp {A}_\psi$. Unimportant positive constants are denoted by $C$,
in particular $C$ may vary from occurrence to occurrence.
The dependence on other variables is  sometimes indicated by
subscripts  such as $C_\nu$.

\section{Discussion and examples}\label{sec:11.7.9.0.25}

In this section we investigate how general our conditions are by looking at several examples.

\subsection{Global conditions}\label{sec:global-conditions-1}
We recall some general criteria for self-adjointness and
the unique continuation property.

\begin{proposition}\label{prop:global-conditionss}
Let $(M,g)$ be a complete Riemannian manifold of dimension $d\geq 1$.
Then the free Schr\"odinger operator $H_0$ is essentially self-adjoint
on $C^\infty_{\mathrm{c}}(M)$. Suppose  $V$ is real-valued,
measurable, bounded outside  a compact
set and in addition: $V\in L^2_{\rm loc }(M)$ for $d=1,2,3$,  $V\in L^p_{\rm
       loc }(M)$ for some $p>2$ if  $d=4$ while $V\in
     L^{d/2}_{\rm loc }(M)$ for $d\geq 5$. Then $V$ is relatively
     compact. In particular $H$ is essentially self-adjoint on
     $C^\infty_{\mathrm{c}}(M)$.
\end{proposition}\label{prop:global-conditions2}
We refer to  \cite{Ch}
and \cite[Theorems X.20 and X.21]{RS}. We can generalize the class of
potentials to the Stummel class, see  e.g.  \cite{DoGa}. 

As for the unique continuation property, Condition
\ref{cond:11.7.9.0.24}, there is an extensive
literature although mostly for Schr\"odinger operator theory, see e.g. \cite{JK}. 
For general connected manifolds
we refer to \cite{Wu} and references therein, quoting here the following sufficient
conditions supplementing connectivity and the conditions in Proposition
\ref{prop:global-conditionss}: 1) $d= 2,3,4$ and  $V$ is globally bounded, or  2) $d\geq
5$.  One could
(of course) add 3)  $d=1$. 

\subsection{Conditions inside an end}\label{sec:conditions-inside-an}
In the sequel we consider a connected and
complete $(M,g)$  of dimension $d\geq 2$ and take (for simplicity) $V=0$. We shall
examine the meaning of Condition \ref{cond:diffeo} in the case where,
in addition, 
$(M,g)$ has the following explicit {\it end}  structure:
There exists an open subset $E\subset M$ such that isometrically the
closure 
$\bar E\cong [0,\infty) \times S$ for some $(d-1)$-dimensional 
manifold $S$,  and that
\begin{align}
g=\mathrm{d}r\otimes \mathrm{d}r+g_{\alpha\beta}(r,\sigma)\,\mathrm{d}\sigma^\alpha\otimes\mathrm{d}\sigma^\beta;
\quad g_{rr}=1,\ g_{r\alpha}=g_{\alpha r}=0,
\label{eq:11.7.25.3.20}
\end{align}
where $(r,\sigma)\in [0,\infty) \times S$ denotes local coordinates and the Greek indices run over $2,\dots,d$.
Whence actually $r$ is globally defined in $E$ and it is a smooth
distance function (here given as the  distance to $\{0\}\times S$).  In particular we have $|\mathrm{d} r|=1$ which
obviously implies (\ref{eq:3}) and 
(\ref{eq:11.7.15.8.1}). Notice here that Condition \ref{cond:diffeo}
involves only the part of the function $r$  at large values, so in
agreement with Condition \ref{cond:diffeo} we can cut and extend it to a
smooth function on $M$ obeying $r\geq 1$. This is tacitly understood
below. To examine the remaining statements
\eqref{eq:10.2.8.3.55} and \eqref{eq:10.9.2.23.19} of Condition
\ref{cond:diffeo} we  compute
\begin{subequations}
\begin{align}
\nabla^2r^2&=2\,\mathrm{d}r\otimes \mathrm{d}r
+r(\partial_r g_{\alpha\beta})\,\mathrm{d}\sigma^\alpha\otimes \mathrm{d}\sigma^\beta,
\label{eq:11.8.5.22.23}\\
\triangle r^2&=g^{ij}(\nabla^2r^2)_{ij}=2+rg^{\alpha\beta}(\partial_rg_{\alpha\beta}).\label{eq:2verif}
\end{align}
  \end{subequations}

\subsubsection{End of warped product type}
If we consider the \textit{warped product} case where 
$g_{\alpha\beta}(r,\sigma)=f(r)h_{\alpha\beta}(\sigma)$
we obtain, using \eqref{eq:11.8.5.22.23} and \eqref{eq:2verif}, the following examples fulfilling  also
\eqref{eq:10.2.8.3.55} and \eqref{eq:10.9.2.23.19}  of Condition~\ref{cond:diffeo}.
\begin{examples}\label{example:8.2.13.11}
\begin{enumerate}[(1)]
\item \label{item:1} Let $f=r^{2p}$ with $p>0$. Then
  \eqref{eq:10.2.8.3.55} and \eqref{eq:10.9.2.23.19}  hold
  with $c=\min\{2,2p\}$ and $\rho_1=0$ respectively, and  the critical energy $E_0=0$.
\item \label{item:2}Let  $f=\mathop{\mathrm{exp}} (\kappa r^{q})$ with
  $\kappa >0$
  and $q\in (0,1)$.  Then \eqref{eq:10.2.8.3.55} and
  \eqref{eq:10.9.2.23.19} hold
  with $c=2$ and $\rho_1=0$ respectively, and $E_0=0$.
\item \label{item:5}Let  $f=\mathop{\mathrm{exp}} (2\kappa r)$ with
  $\kappa >0$.
  Then \eqref{eq:10.2.8.3.55} and \eqref{eq:10.9.2.23.19} hold
  with $c=2$ and $\rho_1=0$ respectively, and $E_0=\kappa^2(d-1)^2/8$.
\end{enumerate}
\end{examples}
\begin{remarks} \label{remark:end-warped-products} 
  \begin{enumerate}[1)]
  \item \label{item:6}
 For all of these examples  it is easy to compute that
  the essential spectrum $\sigma_\ess (H)\supseteq [E_0,
  \infty)$. If in addition $M\setminus E$ and $S$ are  compact
then we have  $\sigma_\ess (H)= [E_0,
  \infty)$. Whence
  indeed  the absence of eigenvalues in $(E_0, \infty)$ as stated in
  Theorem \ref{thm:absence-eigenvalues-1} is optimal under these
  additional conditions for the above   
  examples (except possibly that the threshold energy $E=E_0$ in a
  concrete situation might not
  be an eigenvalue neither).
\item \label{item:7} A metric obtained by taking $p=1$ in \ref{item:1} 
  (and assuming also $M\setminus E$ and
  $S$ compact), and possibly perturb it,  is  dubbed a ``scattering
  metric'' in \cite{Me, MZ}. 
  As shown by Melrose  
  absence of positive eigenvalues holds for scattering
  metrics. Since it is not required in
  Condition~\ref{cond:diffeo} that $r$ is an exact distance function
  we may  still have this condition fulfilled in perturbed
  situations (letting $r$ be the  unperturbed distance function). In
  this spirit Donnelly  
  \cite{Do}  studied perturbations of the
  Euclidean metric (corresponding to $p=1$ in \ref{item:1}) using a
  certain function of this type (i.e. not an exact distance function), and he proved
  absence of positive eigenvalues for such model. More generally, but
  roughly still in the framework of perturbations of \ref{item:1},
  absence of embedded  eigenvalues  was obtained in \cite{Ku2}, and
  for hyperbolic models (roughly for perturbations of \ref{item:5}) it was done
  in \cite{Ku}.  However Kumura's results are stated in terms of an exact
  distance function and parts of his results involve conditions on the
  radial curvature. Whence his framework is seemingly somewhat
  different. It turns out, however, that his conditions imply properties that are stronger
  than 
  our conditions. We will discuss an example of this point in Corollary
  \ref{cor:11.8.5.18.18} and Remark \ref{rem:volume-growth-curv}
  \ref{item:10}.
\item \label{item:8} Under the condition of warped product metrics    growth rates between 
$f=r^{2p}$ with $p>1/2$ and $f=\mathop{\mathrm{exp}} (\kappa r^{q})$ with
$\kappa>0$ and $q\in (0,1/2)$ 
define  a  class of metrics for which the scattering
theory  \cite{IS} applies. More generally 
 Conditions \ref{cond:diffeo}--\ref{cond:11.7.9.0.24} are weaker than the
 conditions used in  \cite{IS}.
 \end{enumerate}
\end{remarks}

\subsubsection{Volume growth and curvature}
Here let us relate the critical energy $E_0$ to geometric quantities.
We continue to assume (\ref{eq:11.7.25.3.20}) in the end $E$ although without warped product structure.
In the coordinates $(r,\sigma)\in [0,\infty) \times S$ used in  (\ref{eq:11.7.25.3.20}) we have 
\begin{align*}
\triangle r^2=2+2r\triangle r,\quad \triangle r=\partial_r \ln \sqrt{\det g},
\end{align*}
so that we can measure  the volume growth  in the radial direction in
terms of the function $\triangle r$.
By (\ref{eq:11.8.5.22.23}) the inequality (\ref{eq:10.2.8.3.55}),
necessarily with $c\leq 2$,  is 
equivalent to
\begin{align}
(r\partial_rg_{\alpha\beta}
-cg_{\alpha\beta})_{\alpha,\beta}\ge 0 \mfor r\geq r_0.
\label{eq:11.8.5.22.45}
\end{align}  In particular the induced metric on the angular manifold 
$S_{\tilde r}=\{x\in \bar E|\, r=\tilde r\}$  grows as a function of
$\tilde r$.
By taking the trace of (\ref{eq:11.8.5.22.45})  assuming here and henceforth $c=2$ and $\tilde c=1$ in
(\ref{eq:10.2.8.3.55}) and \eqref{eq:11.7.15.8.1}, respectively,  we 
obtain
\begin{align*}
r\triangle r\ge (d-1)\mfor r\geq r_0.
\end{align*}
Consider the special case of ``asymptotic   volume growth rate''
\begin{align}
\triangle r=\rho_++o(\tfrac{1}{r});\; \rho_+> 0.\label{eq:11.8.5.11.55}
\end{align}
Then, setting $\rho_2=2+2r\rho_+$ and $\rho_1=\triangle
r^2-\rho_2=o(1)$ in \eqref{eq:10.9.2.23.19}, 
we can write 
$E_0$ in terms of the volume growth rate
\begin{align}\label{eq:11.8.5.23.27}
E_0=\rho_+^2/8.
\end{align}
Next, noting that the radial curvatures $R_{\mathrm{rad}}$ can control the
second fundamental form (by a standard comparison argument, see
e.g. \cite[Remark 1.13]{IS} for a reference) 
we recover  a result from  \cite{Ku} (here slightly extended). 
\begin{corollary}\label{cor:11.8.5.18.18}
Suppose   $(M,g)$  is connected and
complete having   an end  $E$ with metric of the form (\ref{eq:11.7.25.3.20}).
Suppose  there exists $\kappa>0$ such that the radial curvature $R_{\mathrm{rad}}$ satisfies 
\begin{align*}
R_{\mathrm{rad}}=-\parb{\kappa^2+o(\tfrac{1}{r})}g \mbox{ on }S_{r}\;(\text{uniformly
  in } x\in E),
\end{align*}
  and there exists $r_1\geq0$ such that
\begin{align*}
 R_{\mathrm{rad}}\leq 0  \text{ on } S_{\tilde r}\text { for all }
  \tilde r\geq r_1 \mand \nabla^2r\ge 0 \mbox{ on  }S_{r_1}.
\end{align*}
 Then $\sigma_{\mathrm{pp}}(H_0)\cap
 ({\kappa^2(d-1)^2}/{8},\infty)=\emptyset$. 
\end{corollary}
\begin{proof}
We have, cf. \cite[Proposition 2.2]{Ku}, 
\begin{align}
\nabla^2 r_{|S_{r}}=(\kappa+o(\tfrac{1}{r}))(g-\mathrm{d}r\otimes\mathrm{d}r),
\label{eq:11.8.7.12.48}
\end{align}
and thus \eqref{eq:11.8.5.11.55} holds with $\rho_+=\kappa(d-1)$.
Indeed we have \eqref{eq:10.2.8.3.55} with $c=2$, and $E_0=\kappa^2(d-1)^2/8$ by (\ref{eq:11.8.5.23.27}).
The result  follows from
Theorem~\ref{thm:absence-eigenvalues-1}.
\end{proof}
\begin{remarks}\label{rem:volume-growth-curv}
\begin{enumerate}[1)]
\item \label{item:9}The radial curvatures $R_{\mathrm{rad}}$ and $K_{\mathrm{rad}}$ of \cite{IS} and \cite{Ku}, respectively,
are different objects but they contain equivalent information. 
\item \label{item:10}The inequalities (\ref{eq:10.2.8.3.55}) and
  (\ref{eq:10.9.2.23.19}) may  be viewed as  bounds on the minimal and 
the mean curvatures (including   the differential of the latter) of $S_r$, respectively, whereas
\eqref{eq:11.8.7.12.48}  certainly is a uniform asymptotic result
for  all the principal
curvatures.
\end{enumerate}
\end{remarks}

\section{Mourre-type commutator}\label{sec:11.8.1.18.3}
Suppose from this point
Conditions~\ref{cond:diffeo}--\ref{cond:11.7.9.0.24}. As a preliminary
step in the proof of Theorem~\ref{thm:absence-eigenvalues-1} we show
in this section a version of the so-called Mourre estimate. We shall use the Mourre-type commutator with respect to the ``conjugate operator''
\begin{align*}
A=\i[H_0,r^2]
=\tfrac{1}{2}\{(\partial_i r^2)g^{ij}p_j+p_i^* g^{ij}(\partial_j r^2)\}
=rp^r+(p^r)^*r;\quad p^r=-\mathrm{i}\partial^r.
\end{align*}
While not necessarily being self-adjoint this operator is 
certainly symmetric as defined on $C^\infty_{\mathrm{c}}(M)$, and 
that suffices for our applications.
\begin{lemma}\label{lem:09.12.8.11.10}
As a quadratic form on $C_{\mathrm{c}}^\infty(M)$,
\begin{align*}
\mathrm{i}[H,A]
&{}=p_i^*(\nabla^2r^2-\tfrac{1}{2}\rho_1g)^{ij}p_j
+\tfrac{1}{2}(\rho_1H_0+H_0\rho_1)
+\mathrm{i}\alpha^ip_i-\mathrm{i}p_i^*\alpha^i+\beta;\\
\alpha_i&{}=\tfrac{1}{4}(\partial_i\rho_2)+V_1(\partial_i r^2),
\\
\beta&{}=(\triangle r^2)V_1-2r\partial^r V_2.
\end{align*}
\end{lemma}
\begin{proof}
  We note the  commutator  formulas, valid for any $\phi\in C^\infty(M)$, 
\begin{subequations}
\begin{align}  \label{eq:seconDer}-[H_0,[H_0,\phi]]&= p_i^*(\nabla^2\phi)^{ij} p_j-\tfrac 14 (\triangle^2
 \phi),\\
  \label{eq:IMS2}
 p_i^* \phi g^{ij}p_j&= \phi H_0+ H_0 \phi +\tfrac 12 (\triangle
 \phi).
\end{align}
  \end{subequations} 
As for \eqref{eq:seconDer} we refer to
  \cite[Lemma 2.5]{Do}  or  \cite[Corollary
4.2]{IS}.  The lemma follows by first using  \eqref{eq:seconDer} with
$\phi=r^2$ and then \eqref{eq:IMS2} with
$\phi=\tfrac 12  \rho_1$.
\end{proof}
 We introduce for $\sigma\ge 0$ 
\begin{align}\label{eq:Ham_sigma}
H_\sigma=H-\tfrac{\sigma^2}{2}|\mathrm{d} r|^2.
\end{align}  We shall consider  $H_\sigma$ and as an
operator defined on 
$C_{\mathrm{c}}^\infty(M)$ only. We recall the definitions of
$\chi_\nu$ and  $E_0$,  
(\ref{eq:11.7.11.5.14}) and \eqref{eq:E_0}, respectively.
\begin{corollary}\label{cor:10.10.13.15.00}
Let $E\in (E_0,\infty)$.
There exist $\gamma>0$ and $C>0$ such that,
if $\nu\ge 1$ is large, then for any $\sigma\ge 0$,
as quadratic forms on $C_{\mathrm{c}}^\infty(M)$,
\begin{align*}
\chi_\nu\mathrm{i}[H_{\sigma},A]\chi_\nu
\ge {}&\gamma\chi_\nu^2
-C\chi_\nu (H_{\sigma}-E)^2 \chi_\nu.
\end{align*}
\end{corollary}
\begin{proof} We shall use Lemma 3.1 and in  particular the functions
  $\alpha$ and $ \beta$ appearing there.
Choose   constants $c'\in (0,\tilde c)$ and 
 $\gamma>0$ such that for  all large
enough $r\geq 1$
\begin{align}\label{eq:4basis}
  r\partial^r|\mathrm{d} r|^2\geq-\tfrac{2c'+\rho_1}{2}|\mathrm{d}
  r|^2\mand \quad E-V
-\tfrac{\alpha^2}{2(c-\tilde c) c'}\geq \gamma/c'.
\end{align} 
Noting $|\triangle r^2|\le Cr$ for large $r$,
cf. \eqref{eq:10.9.2.23.19},  
we  have for  all large
 $r\geq 1$
\begin{subequations}
\begin{align}
\label{eq:4a}
&\nabla^2r^2-\tfrac{1}{2}\rho_1g\geq (c+c'-\tilde c )g,\\
&\beta-\rho_1V+\rho_1 E \geq -\tfrac\gamma 2,\label{eq:4b}\\
&(c'+\tfrac12 \rho_1)^2\leq \tilde c^2.\label{eq:4c}
\end{align}
 \end{subequations}
Then by using \eqref{eq:4a} and the Cauchy Schwarz inequality we
 obtain for all large $\nu\geq 1$ 
\begin{align}
  \begin{split}\label{eq:11.7.12.5.43}
&\chi_\nu\mathrm{i}[H_{\sigma},A]\chi_\nu
\ge \chi_\nu\Bigl\{
(c'+\tfrac{1}{2}\rho_1)(H_\sigma-E)+(H_\sigma-E)(c'+\tfrac{1}{2}\rho_1)
-\tfrac{\alpha^2}{(c-\tilde c)}\\ 
 &  -(2c'+\rho_1)V +(2c'+\rho_1)(\tfrac 12 \sigma^2|\mathrm{d}
 r|^2+E)+\beta+\sigma^2r\partial^r|\mathrm{d} r|^2 
\Bigr\}\chi_\nu.
\end{split}
\end{align}
By using  in turn \eqref{eq:4basis}, \eqref{eq:4b}   and
\eqref{eq:4c} we obtain with $C:=2\tilde c^2/\gamma $
\begin{align*}
\chi_\nu\mathrm{i}[H_{\sigma},A]\chi_\nu
&\ge \chi_\nu
\Bigl\{2c'E-2c'V-\tfrac{\alpha^2}{(c-\tilde
  c)}-(c'+\tfrac{1}{2}\rho_1)^2/C-C(H_\sigma-E)^2-\tfrac\gamma 2
\Bigr\}\chi_\nu\\
&\ge \chi_\nu
\Bigl\{2\gamma-\tfrac\gamma 2-C(H_\sigma-E)^2-\tfrac
\gamma 2
\Bigr\}\chi_\nu,
\end{align*}
  and whence  the assertion.
\end{proof}

\section{Exponential decay of eigenstates}
The proof of Theorem~\ref{thm:absence-eigenvalues-1}, given in this
section,  depends on the following exponential decay estimate which in
turn will be proved in Section \ref{sec:auxiliary-operators}.

\begin{proposition}\label{prop:absence-eigenvalues-1b}
Let $E\in\sigma_{\pp}(H)\cap (E_0,\infty)$ and suppose 
$\phi\in\vD(H)$ satisfies $H\phi = E\phi$. 
Then for any $\sigma\ge 0$ one has $\e^{\sigma r}\phi\in \vH$.
\end{proposition}
To implement Condition \ref{cond:11.7.11.1.25} efficiently we need 
to strengthen the stated approximation property under some additional
conditions (fulfilled for eigenstates 
 due to Proposition \ref{prop:absence-eigenvalues-1b}). 
\begin{lemma}\label{lem:11.7.19.22.23}
Let $\psi\in {\mathcal D}(H)$. There exists $\nu_0\ge 1$ such that for
$\nu\ge \nu_0$ and for any $\sigma\ge 0$ such that  $\mathrm{e}^{\sigma r}\psi,\mathrm{e}^{\sigma r}H\psi \in {\mathcal H}$
the following properties hold: The states $\chi_\nu\mathrm{e}^{\sigma r}p
\psi,\mathrm{e}^{\sigma r}p \chi_\nu\psi\in{\mathcal H}$ and there
exists a sequence $\psi_n\in C^\infty_{\mathrm{c}}(M)$ (possibly
depending on $\sigma$) such that as $n\to \infty$
\begin{align}
\|\chi_\nu\mathrm{e}^{\sigma r}(\psi-\psi_n)\|
+\|\chi_\nu\mathrm{e}^{\sigma r}(p \psi-p\psi_n)\|
+\|\chi_\nu\mathrm{e}^{\sigma r}(H\psi-H\psi_n)\|\to 0.
\label{eq:11.7.19.12.53}
\end{align}
\end{lemma}
\begin{proof} {\noindent \it  Step I} Note the distributional identity
  \begin{equation*}
   \chi_\nu\mathrm{e}^{\sigma r}p \psi=\mathrm{e}^{\sigma r}p
   \chi_\nu\psi+\i \mathrm{e}^{\sigma r}  \psi\chi_\nu'\d r.
  \end{equation*} Applied to the given $\psi$ we see that
  $\chi_\nu\mathrm{e}^{\sigma r}p \psi \in{\mathcal H}$ if and only if
  $\mathrm{e}^{\sigma r}p \chi_\nu\psi\in{\mathcal H}$.

{\noindent \it  Step II} We  claim that 
there exists $C>0$ such that,
if $\nu\ge 1$ is large, then for any $\psi\in C^\infty_{\mathrm{c}}(M)$
and $\sigma\ge 0$
\begin{align}
\|\chi_\nu\mathrm{e}^{\sigma r}|p \psi|\|^2
\le \|\chi_\nu\mathrm{e}^{\sigma r}H\psi\|^2
+C\langle \sigma\rangle^2\|\chi_{\nu/2}\mathrm{e}^{\sigma r}\psi\|^2.
\label{eq:11.7.22.9.5}
\end{align}
In fact  by  \eqref{eq:IMS2} 
\begin{align*}
\|\chi_\nu\mathrm{e}^{\sigma r}|p\psi|\|^2
&=
2\mathop{\mathrm{Re}}{}\inp{\chi_\nu \mathrm{e}^{\sigma r}\psi, \chi_\nu\mathrm{e}^{\sigma r} H\psi}
+\tfrac{1}{2}\inp{\psi, (\triangle \chi_\nu^2\mathrm{e}^{2\sigma r})\psi}
-2\inp{\chi_\nu\mathrm{e}^{\sigma r}\psi,
V\chi_\nu\mathrm{e}^{\sigma r}\psi}\\
&\le \|\chi_\nu\mathrm{e}^{\sigma r}H\psi\|^2
+C\langle \sigma\rangle^2\|\chi_{\nu/2}\mathrm{e}^{\sigma r}\psi\|^2.
\end{align*} Here we used Condition~\ref{cond:diffeo} and the
following consequence
\begin{align}
|\triangle r|=\tfrac{1}{2r}|(\triangle r^2)-2|\mathrm{d}r|^2|\le
C\mfor r=r(x)\text { large}.
\label{eq:11.7.22.9.52}
\end{align}

{\noindent \it  Step III} We consider the case $\sigma=0$, and hence  suppose only $\psi \in {\mathcal D(H)}$.
Let $\psi_n\in C^\infty_{\mathrm{c}}(M)$ and large $\nu\ge 1$  be as in 
Condition~\ref{cond:11.7.11.1.25}.
Then, regarding   (\ref{eq:11.7.19.12.53}), it suffices to consider the middle term.
By (\ref{eq:11.7.22.9.5}) we have
\begin{align*}
\|\chi_\nu (p\psi_n-p \psi_{n'})\|^2
\le C\parb{\|\chi_\nu (H\psi_n-H\psi_{n'})\|^2
+\|\chi_{\nu/2}(\psi_n-\psi_{n'})\|^2}.
\end{align*}
This implies $\chi_\nu p \psi_n$ converges strongly.
 Since also $\chi_\nu p \psi_n$ converges in  distributional
 sense to $\chi_\nu p \psi$,
we obtain that the limit 
$\chi_\nu p\psi \in {\mathcal H}$ and then in turn, by
letting $n'\to \infty$ above,  
(\ref{eq:11.7.19.12.53}) for $\sigma=0$.

{\noindent \it  Step IV} We let $\sigma>0$ and 
suppose $\mathrm{e}^{\sigma r}\psi,\mathrm{e}^{\sigma r}H\psi \in {\mathcal H}$.
 Choose $\psi_n\in C^\infty_{\mathrm{c}}(M)$ and large $\nu\ge 1$ 
as in Condition~\ref{cond:11.7.11.1.25}, again.
As for the first and the third terms of (\ref{eq:11.7.19.12.53}),
we compute as follows:
Put $\psi_{n,\nu'}=\bar \chi_{\nu'}\psi_n$ for  $\nu'\ge
2\nu$ and with $\bar \chi_{\nu'}:=1-\chi_{\nu'}$. Then we decompose 
\begin{align}
\chi_\nu\mathrm{e}^{\sigma r}(\psi-\psi_{n,\nu'})
=\bar \chi_{\nu'}\mathrm{e}^{\sigma r}\chi_\nu(\psi-\psi_n)
+\chi_{\nu'}\mathrm{e}^{\sigma r}\psi.
\label{eq:11.7.19.20.55}
\end{align}
 We put
\begin{align}\label{eq:1}
  R_{\nu'}=\mathrm{i}[H,\chi_{\nu'}]
=\tfrac{1}{2}(\chi_{\nu'}'p^r+(p^r)^*\chi_{\nu'}')=\chi_{\nu'}'p^r-\tfrac
\i 2 \parb{\chi_{\nu'}''|\d r|^2+\chi_{\nu'}'\triangle r},
\end{align} and  decompose similarly
\begin{align}
\begin{split}
&\chi_\nu\mathrm{e}^{\sigma r}(H\psi-H\psi_{n,\nu'})\\
&=\bar \chi_{\nu'}\mathrm{e}^{\sigma r}\chi_\nu(H\psi-H\psi_n)
+\chi_{\nu'}\mathrm{e}^{\sigma r}H\psi
+\mathrm{i}\mathrm{e}^{\sigma r}R_{\nu'}(\psi-\psi_n)
-\mathrm{i}\mathrm{e}^{\sigma r}R_{\nu'}\psi.
\end{split}
\label{eq:11.7.19.20.56}
\end{align}
The norm of the right-hand side of (\ref{eq:11.7.19.20.55})
 can be arbitrarily small by first letting $\nu'$ be large and then
 $n$ large accordingly (using that $\bar \chi_{\nu'}\mathrm{e}^{\sigma
   r}$ is bounded). Similarly the norm of first three terms on the right-hand side of
 \eqref{eq:11.7.19.20.56} 
 can be arbitrarily small by first letting $\nu'$ be large and then
 $n$ large accordingly (for the third term we use
  Step III, i.e. \eqref{eq:11.7.19.12.53} with $\sigma=0$). It remains to consider the last term on the right-hand side of
 \eqref{eq:11.7.19.20.56}. We claim that
 \begin{equation}
   \label{eq:fourth term}
   \|\mathrm{e}^{\sigma r}R_{\nu'}\psi\|\leq C/\nu'.
 \end{equation} To show this we use again Step III to write
 \begin{equation*}
   \|\chi_{\nu'}'\mathrm{e}^{\sigma r}p\psi\|^2 =
\lim_{m\to \infty} \|\chi_{\nu'}'\mathrm{e}^{\sigma
  r}p\psi_m\|^2. 
 \end{equation*} On the other hand by the derivation of
 \eqref{eq:11.7.22.9.5}
 \begin{equation*}
   \|\chi_{\nu'}'\mathrm{e}^{\sigma
  r}p\psi_m\|^2\leq C\parb{\|\chi_{\nu'}'\mathrm{e}^{\sigma r}H\psi_m\|^2
+\parb{\tfrac {\langle \sigma\rangle}{\nu'}}^2\|\chi_{\nu/2}\bar \chi_{2\nu'}\mathrm{e}^{\sigma r}\psi_m\|^2},
 \end{equation*} and hence we conclude   by  taking the limit that 
\begin{align}\label{eq:42term}
  \begin{split}
 \|\chi_{\nu'}'\mathrm{e}^{\sigma
  r}p\psi\|^2&\leq \parb{\tfrac {C_\sigma}{\nu'}}^2 \parb{\|\chi_\nu\bar \chi_{2\nu'}\mathrm{e}^{\sigma r}H\psi\|^2
+\|\chi_{\nu/2}\bar \chi_{2\nu'}\mathrm{e}^{\sigma r}\psi\|^2}\\ & \leq \parb{\tfrac {C_\sigma}{\nu'}}^2 \parb{\|\mathrm{e}^{\sigma r}H\psi\|^2
+\|\mathrm{e}^{\sigma r}\psi\|^2}.
  \end{split}
\end{align} A consequence of \eqref{eq:42term} is indeed \eqref{eq:fourth
  term}, and whence in turn also  the last term on the right-hand side of
 \eqref{eq:11.7.19.20.56} is small for $\nu'$ sufficiently  large.

We conclude that there exists a sequence of indices $(\nu'(m),n(m))$
so that with 
$\psi_m:=\psi_{n(m),\nu'(m)}$ 
 (here and henceforth  slightly  abusing   notation)
\begin{equation*}
  \|\chi_\nu\mathrm{e}^{\sigma r}(\psi-\psi_m)\|
+\|\chi_\nu\mathrm{e}^{\sigma r}(H\psi-H\psi_m)\|\to 0.
\end{equation*} In particular, using here \eqref{eq:11.7.22.9.5},  the right-hand side of 
\begin{align*}
\|\chi_{2\nu}\mathrm{e}^{\sigma r}p (\psi_n-\psi_{n'})\|^2
\le C\parb{\|\chi_{2\nu}\mathrm{e}^{\sigma r}H(\psi_n-\psi_{n'})\|^2
+\|\chi_{\nu}\mathrm{e}^{\sigma r}(\psi_n-\psi_{n'})\|^2}
\end{align*} is small for $n,n'\to \infty$. We can from
this point  mimic the last part of Step III. 
\end{proof}
\begin{proof}[Proof of Theorem~\ref{thm:absence-eigenvalues-1}]
Suppose $E\in\sigma_{\pp}(H)\cap (E_0,\infty)$ and let $\phi$ be any corresponding eigenstate.
Then, by Proposition~\ref{prop:absence-eigenvalues-1b}, 
for any $\nu\ge 1$ and $\sigma\ge 0$
\begin{align}
\phi_\sigma=\phi_{\sigma,\nu}:=\chi_\nu\e^{\sigma (r-4\nu)}\phi\in \vH.
\label{eq:11.7.3.23.49}
\end{align}
We will choose $\nu\ge 1$ large in agreement with
Lemma~\ref{lem:11.7.19.22.23} with
 $\psi=\phi$.
In the following computations we actually have to first choose 
an approximate sequence for $\phi$ from $C^\infty_{\mathrm{c}}(M)$ and
then take the limits. This can be done by using 
Lemma~\ref{lem:11.7.19.22.23} and the closedness of $H$,  but since 
the verification is rather straightforward
we shall not elaborate on this point.

We compute, putting
$R_\nu=\mathrm{i}[H_0,\chi_\nu]=\mathop{\mathrm{Re}}{}\parb{\chi_\nu'p^r}$
as in \eqref{eq:1},
\begin{equation}\label{eq:19}
\begin{split}
H\phi_\sigma
={}&E\phi_\sigma+ \tfrac{\sigma^2}{2}|\mathrm{d} r|^2\phi_\sigma
-\mathrm{i}\sigma (\mathop{\mathrm{Re}}{}p^r)\phi_\sigma
-\mathrm{i}\mathrm{e}^{\sigma (r-4\nu)} R_\nu\phi.
\end{split}
\end{equation} In particular indeed $\phi_\sigma\in\vD (H)$.
Take  inner product with $\phi_\sigma$ and compute 
\begin{align*}
\inp{H}_{\phi_\sigma}
=\Re\inp{H}_{\phi_\sigma}=
\inp{E+ \tfrac{\sigma^2}{2}|\mathrm{d} r|^2}_{\phi_\sigma}
+\tfrac{\mathrm{i}}{2}\inp{[R_\nu,\chi_\nu\mathrm{e}^{2\sigma (r-4\nu)}]}_{\phi}.
\end{align*} Whence
\begin{align*}
\inp{H}_{\phi_\sigma}
\ge \inp{E+\tfrac{\sigma^2}{2}|\mathrm{d} r|^2}_{\phi_\sigma}-C\langle\sigma\rangle\|\phi\|^2,
\end{align*} 
where $C>0$ does not depend on $\nu$ or $\sigma$ because $r\leq 2\nu$
on $\mathop{\mathrm{supp}} \chi_\nu'$.
On the other hand  if $c'\in(0,\tilde c)$
and $\nu\ge 1$ is large then, cf.   (\ref{eq:11.7.12.5.43}) with $\sigma=0$, 
\begin{equation*}
2c'\inp{H}_{\phi_\sigma}\leq \inp{\i[H,A]}_{\phi_\sigma}
-\mathop{\mathrm{Re}}{}\inp{\rho_1H}_{\phi_\sigma}
+C\|\phi_\sigma\|^2.
\end{equation*} We fix  such $c'$ assuming in addition (for a later
application) 
\begin{equation}
  \label{eq:2'}
  \liminf_{r\to\infty}{}(r\partial^r|\mathrm{d} r|^2+ c'    |\mathrm{d} r|^2)>  0.
\end{equation}

We compute the first and the second terms on the right-hand side.
By \eqref{eq:19} again
\begin{align}
\begin{split}
&\inp{\i[H,A]}_{\phi_\sigma}\\
&=\sigma^2\mathop{\mathrm{Im}}{}\inp{A|\mathrm{d} r|^2}_{\phi_\sigma}
-2\sigma \mathop{\mathrm{Re}}{}\inp{(\mathop{\mathrm{Re}} p^r) A}_{\phi_\sigma}
-2\mathop{\mathrm{Re}}{}\inp{R_\nu \mathrm{e}^{\sigma (r-4\nu)}A\chi_\nu\mathrm{e}^{\sigma (r-4\nu)}}_{\phi},
\end{split}
\label{eq:11.7.13.4.10}
\end{align} while 
\begin{align}
\begin{split}
&-\mathop{\mathrm{Re}}{}\inp{\rho_1H}_{\phi_\sigma}\\
&=-E\inp{\rho_1}_{\phi_\sigma}
-\tfrac{\sigma^2}{2}\inp{\rho_1|\mathrm{d} r|^2}_{\phi_\sigma}
-\sigma\mathop{\mathrm{Im}}{}
\inp{\rho_1\mathop{\mathrm{Re}}{}p^r}_{\phi_\sigma}
-\mathop{\mathrm{Im}}{}
\inp{\rho_1\chi_\nu\mathrm{e}^{2\sigma (r-4\nu)} R_\nu}_\phi
\end{split}
\label{eq:11.8.1.8.16}
\end{align}
The first and the second terms of (\ref{eq:11.7.13.4.10}) are
estimated using 
\begin{align*}
\mathop{\mathrm{Im}}{}(A|\mathrm{d} r|^2)
&=-r(\partial^r|\mathrm{d} r|^2),\\
-2\mathop{\mathrm{Re}}{}((\mathop{\mathrm{Re}} p^r) A)
&=-(\mathop{\mathrm{Re}} p^r) (2r(\mathop{\mathrm{Re}}{}p^r)
-\mathrm{i}|\mathrm{d} r|^2)
+\mathrm{h.c.}
\le (\partial^r|\mathrm{d} r|^2).
\end{align*}
As for the third term of (\ref{eq:11.7.13.4.10}) we estimate
(recall the notation $\bar \chi_\nu=1-\chi_{\nu}$)
\begin{align*}
&-2\mathop{\mathrm{Re}}{}\inp{R_\nu \mathrm{e}^{\sigma (r-4\nu)}
A\chi_\nu\mathrm{e}^{\sigma (r-4\nu)}}_{\phi}\\
&\le \|\mathrm{e}^{\sigma (r-4\nu)}R_\nu\phi\|^2+
\|\bar\chi_{2\nu}A\chi_\nu\mathrm{e}^{\sigma (r-4\nu)}\phi\|^2\\
&\le \bigl\{\|\chi_\nu'\mathrm{e}^{\sigma (r-4\nu)}p^r\phi\|
+\tfrac{1}{2}\|(\chi_\nu''|\mathrm{d} r|^2+\chi'_\nu(\triangle r))
\mathrm{e}^{\sigma (r-4\nu)}\phi\|\bigr\}^2\\
&\phantom{\le {}}+\bigl\{\|2r\bar\chi_{2\nu}\chi_\nu\mathrm{e}^{\sigma (r-4\nu)}p^r \phi\|+\|\bar\chi_{2\nu}(2r|\mathrm{d} r|^2\chi_\nu'+2\sigma r\chi_\nu|\mathrm{d}r|^2
+\tfrac{1}{2}(\triangle r^2)\chi_\nu) \mathrm{e}^{\sigma (r-4\nu)} \phi\|\bigr\}^2\\
&\le C\nu^2\|\chi_{\nu/2}|p \phi|\|^2+C\nu^2\langle \sigma\rangle^2\|\phi\|^2,
\end{align*}
where we have used \eqref{eq:11.7.22.9.52}.
By using \eqref{eq:11.7.19.12.53} and (\ref{eq:11.7.22.9.5}) (both with $\sigma=0$)
we then conclude 
\begin{align*}
-2\mathop{\mathrm{Re}}{}\inp{R_\nu \mathrm{e}^{\sigma (r-4\nu)}A\chi_\nu
\mathrm{e}^{\sigma (r-4\nu)}}_{\phi}
\le C\nu^2\langle\sigma\rangle^2\|\phi\|^2.
\end{align*}
Next, we compute the third and fourth terms of (\ref{eq:11.8.1.8.16}).
Note that we can not differentiate $\rho_1$.
But by the support property of $\chi_\nu'$ (the one used before) the fourth term is estimated similarly to the third term of (\ref{eq:11.7.13.4.10}),
and we obtain
\begin{align*}
-\mathop{\mathrm{Im}}{}
\inp{\rho_1\chi_\nu\mathrm{e}^{2\sigma (r-4\nu)} R_\nu}_\phi\le C\langle\sigma\rangle^2\|\phi\|^2.
\end{align*}
We proceed for the third term of (\ref{eq:11.8.1.8.16}):
\begin{align*}
&-\mathop{\mathrm{Im}}{}
\inp{\rho_1\mathop{\mathrm{Re}}{}p^r}_{\phi_\sigma}\\
&=-\mathop{\mathrm{Im}}{}
\inp{\rho_1p^r}_{\phi_\sigma}
+\tfrac{1}{2}\inp{\rho_1(\triangle r)}_{\phi_\sigma}\\
&\le 
-\mathop{\mathrm{Im}}{}
\inp{\phi_\sigma, \rho_1\chi_\nu\mathrm{e}^{\sigma (r-4\nu)}p^r\phi}
+C\|\phi\|^2
+C\sigma\inp{|\rho_1|}_{\phi_\sigma}
+C\|\phi_\sigma\|^2\\
&\le C\bigl(\sup \chi_{\nu/2}|\rho_1|\bigr)
\langle \sigma\rangle^{-1}\|\chi_\nu\mathrm{e}^{\sigma (r-4\nu)}|p\phi|\|^2
+C\|\phi\|^2
+C\langle\sigma\rangle\inp{|\rho_1|}_{\phi_\sigma}
+C\|\phi_\sigma\|^2.
\end{align*}
We apply \eqref{eq:11.7.19.12.53} and (\ref{eq:11.7.22.9.5}) 
to  the first term on the right-hand side yielding 
\begin{align*}
-\mathop{\mathrm{Im}}{}
\inp{\rho_1\mathop{\mathrm{Re}}{}p^r}_{\phi_\sigma}
\le C\bigl(\sup \chi_{\nu/2}|\rho_1|\bigr)\langle \sigma\rangle\|\phi_\sigma\|^2
+C\langle\sigma\rangle\|\phi\|^2+C\|\phi_\sigma\|^2.
\end{align*}

We  summarize 
\begin{equation}\label{eq:18}
   \sigma^2\inp[\big]{r(\partial^r|\mathrm{d} r|^2)+{c'}|\mathrm{d} r|^2
-C\bigl(\sup\chi_{\nu/2}|\rho_1|\bigr)}_{\phi_\sigma}
  -C\langle\sigma\rangle\|\phi_\sigma\|^2\leq C\nu^2\langle\sigma\rangle^2\|\phi\|^2.
\end{equation} 
We shall apply \eqref{eq:18} to a fixed $\nu\ge 1$  chosen so large  that
the quantity
$r(\partial^r|\mathrm{d} r|^2)+c'|\mathrm{d} r|^2
-C\bigl(\sup\chi_{\nu/2}|\rho_1|\bigr)$ is greater than some positive constant on 
$\mathop{\mathrm{supp}}\chi_\nu$. Note that this in turn is doable
since we have assumed \eqref{eq:2'}. 

Now assume $\chi_{5\nu}\phi\not\equiv 0$.
After division by $\langle\sigma\rangle^2$ on both sides of \eqref{eq:18} the left-hand side grows exponentially as $\sigma\to\infty$
whereas the right-hand side is bounded, 
and hence we obtain a contradiction.
Thus $\chi_{5\nu}\phi\equiv 0$, and then by
Condition~\ref{cond:11.7.9.0.24} we conclude that $\phi(x)= 0$ in $M$. 
\end{proof}

\section{Auxiliary operators}\label{sec:auxiliary-operators}
In this section we give the proof of  Proposition~\ref{prop:absence-eigenvalues-1b}.
We introduce  regularized weights
$$
\tk(r) = r(1+\tfrac{r}{m})^{-1},\quad m\geq 1,
$$
and denote the derivatives in $r$ by $\theta^{(k)}_m(r)$, e.g.,
$$
\tkp(r) = \theta_m^{(1)}(r)= (1+\tfrac{r}{m})^{-2}.
$$
We introduce furthermore
$$
\Tk(r) = \Theta_m^{\sigma,\delta}(r) = \sigma r +\delta
\tk(r),\quad \sigma,\delta\geq 0,
$$
and denote the derivatives by $\Theta_m^{(k)}(r)$ as above.
Now we define some observables:
\begin{align*}
B&=\i[H_0,r]=\tfrac{1}{2}(p^r+(p^r)^*)
=p^r+\tfrac{1}{2\mathrm{i}}(\triangle r),\\
B_m &=\mathrm{i}[H_0,\Theta_m]
=\tfrac1{2}\left(\Theta_m'p^r+ (p^r)^*\Theta_m'\right)
= \Theta'_mp^r
+\tfrac{1}{2\mathrm{i}}\{(\triangle r)\Theta_m'+|\mathrm{d} r|^2\Theta_m''\},\\
R_\nu&=\mathrm{i}[H_0,\chi_{\nu}]
=\tfrac1{2}(\chi_\nu'p^r+ (p^r)^*\chi_\nu'),\quad \nu \geq 1.
\end{align*}
Then we have the properties:
\begin{subequations}
\begin{align}
\label{eq:11.7.14.19.33}
A&=2Br-\tfrac{1}{\mathrm{i}}|\mathrm{d} r|^2
=2rB+\tfrac{1}{\mathrm{i}}|\mathrm{d} r|^2\\
\label{eq:12}
B_m
&= B\Tkp -\tfrac{1}{2\mathrm{i}}|\mathrm{d} r|^2\Theta_m''
= \Tkp B +\tfrac{1}{2\mathrm{i}}|\mathrm{d} r|^2\Theta_m'',\\
\begin{split}
(B_m)^2&=B(\Tkp)^2B 
-\tfrac{1}{2}(\partial^r|\mathrm{d} r|^2)\Theta_m'\Theta_m''
-\tfrac{1}{2}|\mathrm{d} r|^4\Theta_m'\Theta_m'''
-\tfrac{1}{4}|\mathrm{d} r|^4(\Theta_m'')^2\\
&\le B(\Tkp)^2B +C\delta(\sigma+\delta),
\end{split}\label{eq:11}
\end{align}
\end{subequations}
where the last inequality is for large $r$. We  set for $\nu'\ge 2\nu$ and $\psi\in C_{\mathrm{c}}^\infty(M)$
\begin{align*}
\psi_m=\psi_{m,\nu,\nu'}=\chi_{\nu,\nu'}\mathrm{e}^{\Theta_m}\psi;\quad
\chi_{\nu,\nu'}=\chi_{\nu}\bar \chi_{\nu'},\quad \bar \chi_{\nu'}=1-\chi_{\nu'},
\end{align*} not to be mixed up with $\psi_n$ in
Lemma~\ref{lem:11.7.19.22.23}. We recall the notation
\eqref{eq:Ham_sigma}. A computation shows,
 cf. \eqref{eq:19}, that 
\begin{align}
\begin{split}
&\i(H_\sigma-E)  \psi_m\\
&=\i \chi_{\nu,\nu'}\e^\Tk (H-E)\psi
+\left\{B_m -\tfrac{1}{2\mathrm{i}}((\Theta_m')^2-\sigma^2)|\mathrm{d} r|^2\right\}\psi_m
+\mathrm{e}^{\Theta_m}(R_\nu-R_{\nu'})\psi. 
\end{split}\label{eq:13}
\end{align}

\begin{lemma}\label{lemm:absence-eigenvalues-2}
Let $\sigma_0 \geq 0$ be fixed.
\begin{enumerate}[\normalfont (i)]
\item\label{item:3} 
Let $\epsilon>0$.
Then there exists $C>0$ such that, if $\nu\ge 1$
is large, for any $m\ge 1$, $0\le \delta\le 1$ and $0\le \sigma\le \sigma_0$,
as quadratic forms on $C_{\mathrm{c}}^\infty(M)$, 
\begin{align*}
\chi_\nu\mathop{\mathrm{Re}}{} (A B_m)\chi_\nu
&\ge 2 \chi_\nu B r\Theta_m' B\chi_\nu-(\epsilon+C\delta )\chi_\nu^2.
\end{align*}

\item\label{item:4}  
Let $\epsilon'>0$. Then there exists $C>0$ such that, 
if $\nu\ge 1$ is large, 
for any $\nu'\ge 2\nu$, $m\ge 1$, $0\le \delta\le 1$, $0\le \sigma\le \sigma_0$,
$E\in\R$ and $\psi\in C_{\mathrm{c}}^\infty(M)$
\begin{align*}
&   \|(H_\sigma-E)\psi_m\|^2 \\
&\leq  5\|\chi_{\nu,\nu'}\e^{\Theta_m}(H-E)\psi\|^2
+ \epsilon'\inp{ B r\Tkp B }_{\psi_m}+ C\delta \|\psi_m\|^2\\ 
&\phantom{\le {}}
+ C_\nu(\|\chi_{\nu/2}\psi\|^2 + \|\chi_{\nu/2}p \psi\|^2 )
+ C({\nu'})^{-2}(\|\chi_{\nu, 2\nu'}\mathrm{e}^{\Theta_m}\psi\|^2 
+ \|\chi_{\nu, 2\nu'}\mathrm{e}^{\Theta_m}p \psi\|^2 ).
\end{align*}

\end{enumerate}
\end{lemma}
\begin{proof}
\textit{(i)}\ 
By (\ref{eq:11.7.14.19.33}) and (\ref{eq:12})
\begin{align*}
\mathop{\mathrm{Re}}{}(AB_m)
&=\tfrac{1}{2}(2Br - \tfrac{1}{\mathrm{i}}|\mathrm{d} r|^2)
(\Theta_m'B+\tfrac{1}{2\mathrm{i}}|\mathrm{d} r|^2\Theta_m'')
+\mathrm{h.c.}\\
&=B r\Theta_m' B
+\tfrac{1}{2\mathrm{i}}Br|\mathrm{d} r|^2\Theta_m''
- \tfrac{1}{2\mathrm{i}}|\mathrm{d} r|^2\Theta_m'B
+ \tfrac{1}{4}|\mathrm{d} r|^4\Theta_m''
+\mathrm{h.c.}\\
&=2B r\Theta_m' B
-\tfrac{1}{2}
\bigl\{(\partial^r |\mathrm{d} r|^2)(\Theta_m'+r\Theta_m'')
+|\mathrm{d} r|^4 (\Theta_m''+r\Theta_m''')\bigr\}.
\end{align*}
Then by \eqref{eq:3} and (\ref{eq:11.7.15.8.1}) the assertion follows.

\smallskip
\noindent
\textit{(ii)}\ 
By (\ref{eq:13}), (\ref{eq:11}), \eqref{eq:3}  and \eqref{eq:11.7.22.9.52}
\begin{align*}
&\|(H_\sigma-E)\psi_m\|^2\\
& \le 5\|\chi_{\nu,\nu'}\mathrm{e}^{\Theta_m}(H-E)\psi\|^2
+5\inp{(B_m)^2}_{\psi_m}
+\tfrac{5}{4}\|((\Theta_m')^2-\sigma^2)|\mathrm{d} r|^2\psi_m\|^2\\
&\phantom{\le {}}+5\|\mathrm{e}^{\Theta_m}R_\nu\psi\|^2
+5\|\mathrm{e}^{\Theta_m}R_{\nu'}\psi\|^2\\
& \le 5\|\chi_{\nu,\nu'}\mathrm{e}^{\Theta_m}(H-E)\psi\|^2
+5\inp{B(\Theta_m')^2B}_{\psi_m}
+C\delta\|\psi_m\|^2\\
&\phantom{\le {}}+C_\nu(\|\chi_{\nu/2}\psi\|^2+\|\chi_{\nu/2}p \psi\|^2)
+ C({\nu'})^{-2}(\|\chi_{\nu, 2\nu'}\mathrm{e}^{\Theta_m}\psi\|^2 
+ \|\chi_{\nu, 2\nu'}\mathrm{e}^{\Theta_m}p \psi\|^2 ).
\end{align*}
Now choose $\nu\ge 1$ large enough so that $5\Theta'_m\le 5(\sigma_0 +1)\le \epsilon' r$ on 
$\mathop{\mathrm{supp}}\chi_\nu$, and  we are done.
\end{proof}

\begin{proof}[Proof of Proposition~\ref{prop:absence-eigenvalues-1b}]
We let $E$ and $\phi$ be as in the proposition.
Set 
\begin{align*}
\sigma_0=\sup {}\{\sigma\ge 0|\ \mathrm{e}^{\sigma r}\phi\in{\mathcal H}\},
\end{align*}
and assume $\sigma_0<\infty$.
If $\sigma_0>0$ we choose $\sigma\in [0,\sigma_0)$ and a small $\delta>0$ such that 
$\sigma+\delta>\sigma_0$. If $\sigma_0=0$ we set $\sigma=0$ and choose a small $\delta>0$.
 These numbers will be determined more precisely in the following
 arguments. In any case we  have
 $\mathrm{e}^{\sigma r}\phi\in{\mathcal H}$.
We indicate  below the dependence of constants using  subscripts.

Due to  Corollary~\ref{cor:10.10.13.15.00}, for  any $\psi\in C_{\mathrm{c}}^\infty(M)$
\begin{align}
\begin{split}
\|\psi_m\|^2
\le \gamma^{-1}\inp{\mathrm{i}[H_{\sigma},A]}_{\psi_m}
+C_0\|(H_\sigma-E)\psi_m\|^2;\;C_0=C/\gamma.
\end{split}
\label{eq:11.7.23.6.53}
\end{align}
We estimate the right-hand side using Lemma~\ref{lemm:absence-eigenvalues-2}.
For the first term of (\ref{eq:11.7.23.6.53})
 we use (\ref{eq:13}) and 
Lemma~\ref{lemm:absence-eigenvalues-2}\ref{item:3} with
$\epsilon=\tfrac{\gamma}{3}$ estimating
\begin{align*}
&\inp{\mathrm{i}[H_{\sigma},A]}_{\psi_m}\\
&=-\inp{\mathrm{i}(H_\sigma-E)\psi_m,A\psi_m}+\mathrm{h.c.}\\
&=-\inp{\mathrm{i}\chi_{\nu,\nu'}\mathrm{e}^{\Theta_m}(H-E)\psi,A\psi_m}
-\inp{B_m\psi_m,A\psi_m}
+\inp{\tfrac{1}{2\mathrm{i}}|\mathrm{d}r|^2((\Theta_m')^2-\sigma^2)\psi_m,A\psi_m}\\
&\phantom{={}}
-\inp{\mathrm{e}^{\Theta_m}(R_\nu-R_{\nu'})\psi,A\psi_m}+\mathrm{h.c.}\\
&\le 2\|\chi_{\nu,\nu'}\mathrm{e}^{\Theta_m}(H-E)\psi\|\|A\psi_m\|
-2\mathop{\mathrm{Re}}{}\inp{AB_m}_{\psi_m}
-\inp{(r\partial^r|\mathrm{d} r|^2)((\Theta_m')^2-\sigma^2)}_{\psi_m}\\
&\phantom{={}}
-\inp{2r|\mathrm{d} r|^4\Theta_m'\Theta_m''}_{\psi_m}
+C_{\nu}(\|\chi_{\nu/2}\psi\|^2+\|\chi_{\nu/2}p\psi\|^2)\\
&\phantom{={}}+ {C_m}(\|\sqrt{r/\nu'}\chi_{\nu,2\nu'}\mathrm{e}^{\sigma r}\psi\|^2
 + \|\sqrt{r/\nu'}\chi_{\nu,2\nu'}\mathrm{e}^{\sigma r}p \psi\|^2 )\\
&\le C(\nu')^2\|\chi_{\nu,\nu'}\mathrm{e}^{\Theta_m}(H-E)\psi\|^2
-4\inp{Br\Theta_m'B}_{\psi_m}
+(\tfrac{2\gamma}{3}+C_1\delta)\|\psi_m\|^2
\\
&\phantom{={}}
+C_{\nu}(\|\chi_{\nu/2}\psi\|^2+\|\chi_{\nu/2}p\psi\|^2)
+ {C_m}(\|\sqrt{r/\nu'}\chi_{\nu,2\nu'}\mathrm{e}^{\sigma r}\psi\|^2
 + \|\sqrt{r/\nu'}\chi_{\nu,2\nu'}\mathrm{e}^{\sigma r}p \psi\|^2 ),
\end{align*} where we used that $r/\nu'\leq 2 \sqrt{r/\nu'}$ on $\supp
\chi_{\nu,2\nu'}$ to estimate $(\nu')^{-2}\|A\psi_m\|^2$.


On the other hand, for the second term of (\ref{eq:11.7.23.6.53}),
let us choose $\epsilon'=\tfrac{4}{\gamma C_0}$ in Lemma~\ref{lemm:absence-eigenvalues-2}\ref{item:4}.
Then (\ref{eq:11.7.23.6.53}) is estimated as 
\begin{align*}
&\|\psi_m\|^2
 \le C(\nu')^2\|\chi_{\nu,\nu'}\mathrm{e}^{\Theta_m}(H-E)\psi\|^2
+\Bigl(\tfrac{2}{3}+(\tfrac{C_1}{\gamma}+C_2)\delta\Bigr)\|\psi_m\|^2
\\
&\quad 
+C_{\nu}(\|\chi_{\nu/2}\psi\|^2+\|\chi_{\nu/2}p\psi\|^2)
+ {C_m}(\|\sqrt{r/\nu'}\chi_{\nu,2\nu'}\mathrm{e}^{\sigma r}\psi\|^2
 + \|\sqrt{r/\nu'}\chi_{\nu,2\nu'}\mathrm{e}^{\sigma r}p \psi\|^2 ).
\end{align*}

Now fix $\nu\ge 1$ sufficiently large (so that the above estimates  hold),
and let $\sigma$ and $\delta$ be such that 
$\tfrac{2}{3}+(\tfrac{C_1}{\gamma}+C_2)\delta\leq \tfrac{3}{4}$
and $\sigma+\delta>\sigma_0$. Then  
\begin{align}
\begin{split}
\tfrac 14\|\psi_m\|^2
&\le C(\nu')^2\|\chi_{\nu,\nu'}\mathrm{e}^{\Theta_m}(H-E)\psi\|^2
+C_{\nu}(\|\chi_{\nu/2}\psi\|^2+\|\chi_{\nu/2}p\psi\|^2)\\
&\phantom{={}}
+ {C_m}(\|\sqrt{r/\nu'}\chi_{\nu,2\nu'}\mathrm{e}^{\sigma r}\psi\|^2
 + \|\sqrt{r/\nu'}\chi_{\nu,2\nu'}\mathrm{e}^{\sigma r}p \psi\|^2 ).
\end{split}
\label{eq:11.7.16.3.22}
\end{align}
By Lemma~\ref{lem:11.7.19.22.23} we can replace $\psi$ of
(\ref{eq:11.7.16.3.22}) by $\phi$. This makes the first term on the
right-hand side 
disappear. Next let
$\nu'\to\infty$ invoking Lebesgue's dominated convergence
theorem. Note that the
third  term 
disappears, and consequently we are left with the bound
\begin{align}
\|\chi_\nu\mathrm{e}^{\Theta_m}\phi\|^2
 &\le 
4C_\nu(\|\chi_{\nu/2}\phi\|^2+\|\chi_{\nu/2}p\phi\|^2).
\label{eq:11.7.16.3.43}
\end{align}
 By letting $m\to\infty$ in \eqref{eq:11.7.16.3.43} invoking Lebesgue's monotone  convergence
theorem we conclude that
$\chi_\nu\mathrm{e}^{(\sigma+\delta)r}\phi\in{\mathcal H}$. This is a 
 contradiction since $\sigma+\delta>\sigma_0$.
\end{proof}

\end{document}